\documentclass[bj,noinfoline,a4paper]{imsart}

\usepackage{amssymb,amsmath,amsthm,mathtools}
\usepackage[numbers]{natbib}
\usepackage[mathcal]{euscript}
\setattribute{journal}{name}{}

\usepackage{bbold}
\def\Isymb{{\mathbb 1}}

\makeatletter

\newcommand{\I}[1][]{\Isymb
        \@ifempty{#1}{\isparamm\subscr\subnozjel}{\subnozjel{#1}}} 
\def\subscr#1{_{(#1)}}
\def\subnozjel#1{_{#1}}

\@ifundefined{texexp}{\let\texexp\exp\def\exp{\texexp\isparam\cbrstar}}\relax

\newcommand{\cond}[1][]{\,#1|\,}
 
\def\cbrstar{\cbr*}

\def\isparam#1{\isparamm#1\relax}

\def\isparamm{\@ifnextchar{\bgroup}}

\newcounter{enumn}\newcounter{enumr}

\newenvironment{rlist}[1][]{%
  \begin{list}{%
      \hbox to 0pt{%
        \hss(\theenumr)%
      }\unskip\ignorespaces
    }{%
      \usecounter{enumr}%
    }%
    \def\theenumr{\roman{enumr}#1}%
  }{%
  \end{list}%
}

\newenvironment{nlist}[1][]{\begin{list}{\hbox to 0pt{\hss(\theenumn)}\unskip\ignorespaces}{\usecounter{enumn}}\def\theenumn{\arabic{enumn}#1}}{\end{list}}

\DeclarePairedDelimiter\abs{\lvert}{\rvert}
\DeclarePairedDelimiter\cbr{\{}{\}}
\DeclarePairedDelimiter\norm{\|}{\|}
\DeclarePairedDelimiter\br{\lbrack}{\rbrack}
\DeclarePairedDelimiter\zjel{(}{)}
\DeclarePairedDelimiter\q{\langle}{\rangle}

\def\PEzjel{\zjel*}
\def\PEbr#1{\br*}

\def\PEfont{\mathbf}
\newcommand{\PE}[1][]{\PEfont{\Pe}_{#1}%
  \@ifnextchar^{\xPE}{\isparam\PEzjel}}
\def\xPE^#1{^{#1}\isparam\PEzjel}
\newcommand{\E}{\def\Pe{E}\PE}
\renewcommand{\P}{\def\Pe{P}\PE}
\newcommand{\PQ}{\def\Pe{Q}\PE}

\def\F{{\mathcal F}}
\def\G{{\mathcal G}}

\let\event=\cbr

\def\set@internal#1#2#3{#1{#2\,:\,#3}}
\def\set{\@ifstar{\set@internal{\cbr*}}{\set@internal\cbr}}

\def\interior{\operatorname{int}}

\def\eps{\varepsilon}

\def\real{{\mathbb R}}

\def\F{{\mathcal F}}
\def\cO{{\mathcal{O}}}

\def\@bs{\relax}
\def\stripbs#1#2\relax#3{%
  \def\next{#1}%
  \ifx\@bs\next\relax\expandafter\@firstoftwo\else\expandafter\@secondoftwo\fi
  {\def#3{#2}}{\def#3{#1#2}}%
}

\def\setbs#1#2\relax{\def\@bs{#1}}
\expandafter\stripbs\string\ \relax\@bs
\expandafter\setbs\@bs\relax\relax

\def\defname#1#2#3#4{%
  \expandafter\stripbs\string#1\relax\nn
  \edef\nn{\csname#3\nn#4\endcsname}%
  \expandafter\def\nn#2\relax%
}

\def\defbx#1{%
  \defname{#1}{{\bar{#1}}}{b}{}%
}

\def\deftx#1{%
  \defname{#1}{{\tilde{#1}}}{t}{}%
}

\newcommand\defxn[2]{%
  \expandafter\newcommand\csname#1n\endcsname[1][n]{#2^{(##1)}}}

\def\ito/{It\^o}

\def\supp{\operatorname{supp}}
\def\conv{\operatorname{conv}}

\defxn SS
\defxn WW
\defxn g\gamma
\defxn s\sigma

\deftx X \deftx S
\deftx R
\defbx\eps\defbx L

\makeatother

\usepackage[colorlinks,citecolor=blue]{hyperref}

\def\d{\mathrm{d}}

\newtheorem{theorem}{Theorem}
\newtheorem{proposition}[theorem]{Proposition}
\newtheorem{lemma}[theorem]{Lemma}
\newtheorem{corollary}[theorem]{Corollary}
\newtheorem*{theorem*}{Theorem}

\theoremstyle{definition}
\newtheorem{df}{Definition}

\makeatletter
\def\barxiv#1{{\url@fmt{ArXiv e-print:~}%
    {\upshape\ttfamily}{#1}{\arxiv@base#1}}}
\def\bdoi#1{{\url@fmt{doi:~}%
    {\upshape\ttfamily}{#1}{\doi@base#1}}}
\makeatother

\begin{document}

\begin{frontmatter}
  \title{Diversity and no arbitrage}
  \runtitle{Diversity and no arbitrage}
  \begin{aug}
    \author{\fnms{Attila}
      \snm{Herczegh}\thanksref{t1,m1}\ead[label=e1]{prince@cs.elte.hu}},
    \author{\fnms{Vilmos}
      \snm{Prokaj}\corref{}\ead[label=e2]{prokaj@cs.elte.hu}\thanksref{t1,m1}}
    \and
    \author{\fnms{Mikl\'os}
      \snm{R\'asonyi}\thanksref{t2,m2}\ead[label=e3]{rasonyi@renyi.hu}}
    \thankstext{t1}{The European Union and the European Social
      Fund have provided financial support to the project under the
      grant agreement no. T\'AMOP 4.2.1./B-09/1/KMR-2010-0003.}
    \thankstext{t2}{On leave from R\'enyi Institute, Budapest.}
    \affiliation{E\"otv\"os Lor\'and University\thanksmark{m1} and 
      MTA Alfr\'ed R\'enyi Institute of Mathematics and University of
      Edinburgh\thanksmark{m2}} 
    \address{\thanksmark{m1}E\"otv\"os Lor\'and University,
      Department of Probability and Statistics\\
      P\'azm\'any P\'eter s\'et\'any 1/C, 
      Budapest, Hungary\\
      \printead{e1}\\
      \noindent\phantom{E-mail: }\printead*{e2}}
    \address{\thanksmark{m2}MTA Alfr\'ed R\'enyi Institute of
      Mathematics, Budapest and University of
      Edinburgh, 
      School of Mathematics\\
      \printead{e3}}
    \runauthor{A. Herczegh, V. Prokaj and M. R\'asonyi}
  \end{aug}
  \begin{abstract}
    A stock market is called diverse if no stock can dominate the market in
    terms of relative capitalization. On one hand, this natural property leads
    to arbitrage in diffusion models under mild assumptions. On the other hand,
    it is also easy to construct diffusion models which are both diverse and
    free  of arbitrage. Can one tell whether an observed diverse market admits
    arbitrage? 
    
    In the present paper we argue that this may well be impossible by proving
    that the known examples of diverse markets in the literature (which do admit
    arbitrage) can be approximated uniformly (on the logarithmic scale) by
    models which are both diverse and arbitrage-free. 
  \end{abstract}

\end{frontmatter}

\section{Introduction}

Stochastic portfolio theory is a relatively new branch of mathematical finance.
It was introduced and studied by \citet{MR1684221,MR1861997}, and then further
developed by \citet{MR2210925}. It provides a framework for analysing portfolio
performance under an angle which is different from the usual one. 

One of the most important notions here is \emph{diversity} of a market. In
short, diversity means that no single stock is ever allowed to dominate the
market. Diversity was proposed based on empirical grounds and is conform with
intuition. 

Absence of arbitrage (riskless profit) is the cornerstone of modern
mathematical finance.  At the technical level, there are various formulations
of arbitrage but basic economic considerations forbid that such opportunities
persist in a liquid market.  

If the log-prices follow an It\^o process with uniformly non-degenerate
voltility matrix, diversity of a market implies the existence of arbitrage
opportunities relative to the market portfolio (see section 7 of
\citet{Karatzas200989}), and thus the non-existence of equivalent martingale
measure also follows (Proposition 6.2 of \citet{Karatzas200989}). 

This situation may seem dramatic at first sight: the common sense notion of
diversity contradicting the most fundamental principle of asset pricing. There
must be a way out: indeed, relaxing the hypothesis of uniformly nondegenerate 
volatility one may easily construct models where both diversity and absence of
arbitrage hold true. However, a much harder question immediately arises: can
we tell whether the price processes seen in today's market (which clearly
satisfy the diversity assumption) are arbitrage-free or not?  

In this paper we derive the somewhat unsettling conclusion that
possibly there is no way to answer this question based on statistical
analysis. Our conclusions parallel those of \cite{GR}.

We are looking at the diverse market models of \citet{MR2210925} and
\citet{Osterrieder2006}. We prove that under an arbitrarily small model
misspecification diversity is retained but relative arbitrage is not. More
precisely, we show that to these diverse market models (admitting relative
arbitrage) there are models arbitrarily close on the logarithmic scale that  
no longer admit arbitrage (though they are still diverse).

The rest of the paper is organized as follows. Section 2 introduces the notion
of diversity presenting examples with and without relative arbitrage
opportunities. Section 3 contains the result on consistent price system that
is needed in subsection 2.3. Section 4 provides a new result on conditional
full support in higher dimensions, an extension of the work of
\citet{MR2731340}. 

\section{Diversity}

Let $T>0$ be a fixed time horizon. We consider a filtered probability space
$(\Omega,(\F_t)_{t\in[0,T]},\P),$ where the filtration is assumed to satisfy
the usual conditions with $\F_0$ being trivial and all events belong to
$\F_T$. 
 
\citet{MR2210925}, see also \citet{Karatzas200989}, call a market diverse
``if no single stock is ever allowed to dominate the entire market in terms of
relative capitalization''. 
To give a formulation of this requirement let the positive processes $S_i$,
$i=1,\dots,n$, denote the capitalization of the $i^{th}$ company.  The market
weights of the companies are defined by 
\begin{displaymath}
  \mu_i(t)=\frac{S_i(t)}{\sum_{j=1}^n S_j(t)}
\end{displaymath}
and we let $\mu_{(1)}(t)=\max_{j}\mu_j(t)$ the largest market
weight. 

A market is called \textit{diverse} on the time-horizon
$[0,T]$ if there exists $\delta\in(0,1)$ such that
\begin{equation*}
  \mu_{(1)}(t)<1-\delta,\quad\text{almost surely for all $t\in[0,T]$}. 
\end{equation*}
Similarly a market is called \textit{weakly diverse} on the time-horizon
$[0,T]$ if for some $\delta\in(0,1)$  
\begin{equation*}
  \frac1T\int_0^T \mu_{(1)}(t)\d t <1-\delta,\quad\text{almost surely}. 
\end{equation*}

A portfolio process $\pi'(t)=(\pi_1(t),\dots,\pi_n(t))$ describes the proportion
of wealth invested in the stocks. It is required that $\pi$ is progressively
measurable and $\pi_i(t)\geq 0$ for $t\in[0,T]$, $i=1,\dots,n$ and
$\sum_i\pi_i(t)=1$ for all $t\in[0,T]$. An example of a portfolio process is
the \textit{market portfolio} defined by the market weights $\mu$.

\citet{Karatzas200989} consider markets where the evolution of the
prices are \ito/ processes, written on the logarithmic scale as
\begin{equation}
  \label{eq:dlogS}
  \d\log S_i(t)=\gamma_i(t)\d t+\sum_{\nu=1}^d \sigma_{i\nu}(t) \d W_{\nu}(t), \quad
  i=1,\dots,n,
\end{equation}
where $W$ is a $d$--dimensional  Brownian motion in the filtration $\F$ and
the coefficients $\gamma,\sigma$ are progressively measurable and satisfy
the integrability condition
$\int_0^T\abs{\gamma(t)}+\norm{\sigma(t)}^2dt<\infty$.  

The value $V^{z,\pi}$ of a portfolio $\pi$ with initial value $z$ is given by
\begin{equation*}
  \frac{\d V^{z,\pi}(t)}{V^{z,\pi}(t)}=
  \sum_{i=1}^{n} \pi_i(t)\frac{\d S_i(t)}{S_i(t)},\quad V^{z,\pi}(0)=z.
\end{equation*}

  Given two portfolios $\pi$ and $\rho$, we say that 
  $\pi$ represents an \textit{arbitrage} opportunity relative to $\rho$ over the
  time-horizon $[0,T]$ if we have $V^{\pi}(0)=V^{\rho}(0)>0$ and
  $$
  \P{V^{\pi}(T)\geq V^{\rho}(T)}=1 \quad\text{and}\quad
  \P{V^{\pi}(T)>V^{\rho}(T)}>0.
  $$

It is an interesting property of  diverse market models that 
there exists arbitrage relative to the market portfolio $\mu$ provided that
there exist $\eps,M>0$ finite constants such that
\begin{equation}
  \label{eq:non-deg}
  \eps \abs{\xi}^2\leq \abs{\sigma'(t)\xi}^2\leq M\abs{\xi}^2,\quad
  \text{a.s. for all $t\in[0,T]$ and $\xi\in\real^n$}.
\end{equation}
Roughly speaking, 
$a(t)=\sigma(t)\sigma'(t)$ is bounded
and non-degenerate uniformly in $(t,\omega)\in[0,T]\times\Omega$.

For the proof of this claim, we refer the reader to
\cite{MR2210925,Karatzas200989}. Note that the existence of relative arbitrage
oppotunity excludes the possibility of the existence of an equivalent
martingale measure, although equivalent local martingale measure may exist.

\subsection{Examples of diverse market}
\label{sec:ex}

We recall in this subsection two examples of diverse markets. The first one is
due to \citet[Theorem 6.1]{MR2210925}. In this type of example the drift is
positive for all but the largest company. The drift of the largest price has a log-pole-type
singularity that prevents its market weights to reach $1-\delta$. 

In the simplest such example, that \citet{MR2210925} present, 
the evolution of the price is written as in \eqref{eq:dlogS}.  The
volatility matrix $\sigma$ satisfies \eqref{eq:non-deg}.
The crucial assumption is about the drift $\gamma$. They fix a vector
$g=(g_1,\dots,g_n)$ of positive numbers. Then $\gamma$ is expressed as
\begin{equation}
  \label{eq:gammai}
  \gamma_i(t)=\I{\mu(t)\notin \cO_i}g_i+
  \I{\mu(t)\in \cO_i}\frac{M}{\delta(\log(\mu_{(1)}(t))-\log(1-\delta))}, 
\end{equation}
where $\cO_i=\set{x\in\real^n}{\max_{j< i}x_j<x_i,\,\max_{j>i}x_j\leq
  x_i}$. Then $\gamma_i$ is $g_i$ except when the $i$-th company has the largest
market weight. In the latter case $\gamma_i$ is negative and decreases  to
$-\infty$ as $\mu_{(1)}(t)$ approchaes $1-\delta$. This negative drift is
strong enough to make the market diverse, that is, it keeps the process
$S(t)$ in the open set
\begin{equation}
  \label{eq:cO def}
  \cO=\cO(\delta)=\set*{x\in(0,\infty)^n}{\max_j \frac{x_j}{x_1+\cdots+x_n}<1-\delta}.
\end{equation}

\citet{Osterrieder2006} concerns arbitrage possibilities 
of diverse markets, in the usual sense, i.e. in the sense of
\citet{MR2200584}.
They give a general construction of diverse markets by
conditioning the price process to stay in $\cO$ for the entire time-horizon
$[0,T]$. They use a condition called \textbf{ND}, staying for
non-degeneracy, that ensures that arbitrage possibilities exist in the diverse
market constructed.  

To be more precise and concrete,
one can start with a pre-model under some probability
$\P[0]$. We may assume, as \citet{Osterrieder2006} do, that  
under  the probability
$\P[0]$ the price processes are  
positive continuous local martingales, that is, 
\begin{equation*}
\frac{\d S_i(t)}{S_i(t)}=\d M_i(t),\quad 1\leq i\leq n,\quad t\geq0.  
\end{equation*}
where $M$ is a continuous local martingale under $\P[0]$. 

Then $\P$ is obtained by conditioning
\begin{equation*}
  \P{A}=\P[0]{A\cond\forall t\in[0,T],\,S(t)\in\cO},\quad\text{for $A\in \F_T$}.
\end{equation*}

We apply this construction with
a special form of $M$, namely with
\begin{displaymath}
  \d M_i(t)= \sum_{\nu=1}^d \sigma_{i\nu}(t) \d W_{\nu}(t),
\end{displaymath}
where the volatility matrix $\sigma$ satisfies \eqref{eq:non-deg} under
$\P[0]$. This condition implies that \textbf{ND} holds on sufficiently small
time-horizons so arbitrage in the usual sense exists under $\P$. Also, there
is arbitrage relative to the market portfolio as \eqref{eq:non-deg} holds
under $\P$.

\subsection{Diverse market models without relative arbitrage}

The market with two assets 
\begin{displaymath}
S_1(t)=\exp{W_1(t)}\quad\text{and}\quad
S_2(t)=\exp{W_1(t)+\arctan({W}_2(t))}  
\end{displaymath}
 (driven by the $2$-dimensional
Brownian motion $W$) is clearly diverse and admits an equivalent martingale
measure at the same time. Note, however, that the volatility of
$(\log S_1,\log S_2)'$
is \emph{not} uniformly non-degenerate, which was an important hypothesis for
showing the existence of relative arbitrage.

Our goal now is to show that in many cases, especially in the examples
recalled in the previous subsection, even though diverse markets
present relative arbitrage opportunities, small model misspecifications or
proportional transaction costs lead to diverse models that no longer admit
arbitrage. 

To state the main theorem of the paper we need the following variant of the
notion of conditional full support.
\begin{df}\label{df:cfs}
  Let $\cO\subset\real^n$ be open set and
  $(S(t))_{t\in[0,T]}$ be a continuous adapted 
  process taking values in $\cO$. We say that $S$ has 
  conditional full support in $\cO$ 
  if for all
  $t\in[0,T]$ and open set $G\subset C([0,T],\cO)$ 
  \begin{equation}
    \label{def:CFS}
    \P{S\in G\cond \F_t}>0,\quad
    \text{a.s. on the event $S|_{[0,t]}\in\set*{g|_{[0,t]}}{g\in G}$}.
  \end{equation}
\end{df}
We will also say that $S$ has full support in $\cO$, or simply full support
when $\cO=\real^n$,  if \eqref{def:CFS} holds
for $t=0$ and for all open subset of $C([0,T],\cO)$. 

Recall  also, the notion of consistent price system. 
\begin{df}
  Let $\eps>0.$ An $\eps$-consistent price system to $S$ is a pair $(\tS,\PQ)$,
  where $\PQ$ is a probability measure equivalent to $\P$ and  $\tS$ is a
  $\PQ$-martingale in the filtration $\F$,  
  such that
  \begin{equation*}
    \frac{1}{1+\eps}\leq \frac{\tS_i(t)}{S_i (t)}\leq 1+\eps ,
    \quad\text{almost surely for all $t\in[0,T]$ and $i=1,\dots,n$}.
  \end{equation*}
\end{df}

Note, that $\tS$ is a martingale under $\PQ$, hence we may assume that it is
c\`adl\`ag, but it is not required in the definition that $\tS$ is continuous.

\begin{theorem}\label{thm:CPS}
  Let $\cO\subset(0,\infty)^n$ be the open set defined by \eqref{eq:cO def} and 
  assume that the price process takes values and has  conditional full support
  in $\cO$. 

  Then for any $\eps>0$ there is an $\eps$-consistent price system  $(\tS,\PQ)$
  such that $\tS$ takes values in $\cO$.
\end{theorem}

The proof of Theorem \ref{thm:CPS} is given in Section \ref{sec:CPS}. 
In the rest of this subsection we show that the examples recalled in
subsection \ref{sec:ex} have conditional full support in $\cO$. Then Theorem
\ref{thm:CPS} 
applies and we can conclude that for any $\eps>0$ there is a
price process $\tS$, uniformly $\eps$-close to $S$ on the logarithmic scale,
such that  no arbitrage (absolute or relative to the market
portfolio) possibilities exist for the price $\tS$.

To check the condition of Theorem \ref{thm:CPS} we apply the next
Theorem whose proof is given in section \ref{sec:CFS}.
To compare it with existing results we mention
that it seems to be new in the sense, that we do not assume that our process
solves a stochastic differential equation as it is done in \citet{MR2190038}
and it is not only for one dimensional processes as it is in
\citet{MR2731340}.

\begin{theorem}\label{thm:CFS}
  Let $X$ be a $n$-dimensional \ito/ process on $[0,T]$, such that 
  \begin{displaymath}
    \d X_i(t)=\mu_i(t)\d t+\sum_{\nu=1}^n\sigma_{i\nu}(t) \d W_\nu(t)
  \end{displaymath}
  Assume that $\abs{\mu}$ is bounded and $\sigma$ satisfies \eqref{eq:non-deg}.
  Then   $X $ has conditional full support.  
\end{theorem}

Consider first the example of diverse market due  to \citet{MR2210925}, see
also the review paper \cite{Karatzas200989}, recalled in subsection
\ref{sec:ex}. So fix a $\delta\in(0,1)$ such that $\cO(\delta)$ is not empty,
and take the coefficients $\gamma$ 
defined in \eqref{eq:gammai}, with $M$ taken from \eqref{eq:non-deg}.

Then we take the open sets $\cO_k=\cO(\delta+1/k)$ 
for $k\geq 1$ and 
note that for any open set $G\subset C([0,T],\cO)$ we have
\begin{equation*}
  G=\bigcup_{k=1}^{\infty} G_k,
\end{equation*}
where $G_k=G\cap C([0,T],\cO_k)$.
Hence it is enough to show that for $t\in[0,T]$ 
\begin{equation}
  \label{eq:S in G_k}
  \P{S\in G\cond\F_t}>0,\quad\text{on $S|_{[0,t]}\in
      \set*{g|_{[0,t]}}{g\in G_k}$}. 
\end{equation}
Let $\tau_k=\inf\set{t\in[0,T]}{S(t)\notin\cO_k}$. Then $\tau_k=\infty$
exactly when $S\in C([0,T],\cO_k)$ while $\tau_k=0$ when $S(0)\notin \cO_k$. 

For $k=1,2,\dots$  
we define the process $\Sn[k]$ with  the equation
\begin{align*}
  \d \Sn[k](t)&= \gn[k](t)\d t + \sigma(t)\d W,\quad \Sn[k](0)=S(0),\\
\intertext{where}
\gn[k](t)&=\gamma(t\wedge \tau_k)\I{S(0)\in\cO_k}.
\end{align*}
Note that $\Sn[k]$ satisfies the conditions of Theorem \ref{thm:CFS}, hence
$\Sn[k]$ has conditional full support (in $\real^n$), and $S=\Sn[k]$ on the
event $\tau_k=\infty$.  

The conditional full support property of $\Sn[k]$ gives, for the open set
$G_k=G\cap C([0,T],\cO_k)$, that for $t\in[0,T]$
\begin{equation*}
  \P{\Sn[k]\in G_k\cond \F_t}>0,\text{a.s. on $\Sn[k]|_{[0,t]}\in\set*{g|_{[0,t]}}{g\in
      G_k}$}.  
\end{equation*}
To obtain \eqref{eq:S in G_k} one has to add only that $\event{\Sn[k]\in
  G_k}=\event{S\in G_k}\subset \event{S\in G}$. This proves that Theorem
\ref{thm:CPS} applies to the diverse market constructed by \citet{MR2210925}.

Next we turn to diverse market model attributed to
\citeauthor{Osterrieder2006} in subsection \ref{sec:ex}.
By Theorem \ref{thm:CFS} the process $S$ has conditional full support under
$\P[0]$.  For an open set $G\subset C([0,T],\cO)$ we have by Bayes
formula
\begin{displaymath}
  \P{S\in G\cond \F_t}=
  \frac{\E[0]{\I{S\in G}\frac{\d\P}{\d\P[0]}\cond \F_t}}
  {\E[0]{\frac{\d\P}{\d\P[0]}\cond \F_t}}=
  \frac{\P[0]{S\in G\cond \F_t}}
  {\P[0]{S\in C([0,T],\cO)\cond \F_t}}.
\end{displaymath}
Since $S$ has conditional full support under $\P[0]$ both the numerator and
the denominator are positive on the event $S|_{[0,t]}\in\set{g|_{[0,t]}}{g\in
  G}$. Then
\begin{displaymath}
  \P{S\in G\cond \F_t}>0,\quad\text{on $S|_{[0,t]}\in\set*{g|_{[0,t]}}{g\in
  G}$}.
\end{displaymath}
So $S$ has conditional full support in $\cO$ under
the measure $\P$ and Theorem \ref{thm:CPS} applies to this type of examples as
well.

\section{Consistent Price System and 
Conditional Full support} 
\label{sec:CPS}
  
The aim of this section is to prove Theorem \ref{thm:CPS}. It will follow from
the following reinforcement of a result due to \citet{MR2398764}.

\begin{theorem}\label{thm:cps}
  Let $\cO\subset\real^n$ be an open set and  $(S(t))_{t\in[0,T]}$ be an
  $\cO$--valued, continuous adapted process having conditional full
  support  in $\cO$.

  Besides, let $(\eps_t)_{t\in[0,T]}$ be a continuous positive process, that
  satisfies 
  \begin{equation}\label{eq:eps-LS}
    \abs{\eps_t-\eps_s}\leq L_s 
    \sup_{s\leq u\leq t}\abs*{S(u)-S(s)},\quad\text{for all $0\leq s\leq
      t\leq T$} 
  \end{equation}
  with some progressively measurable  finite valued $(L_s)_{s\in[0,T]}$.

  Then  $S$ admits an  $\eps$-consistent price system in the sense that, 
  there is an equivalent probability $\PQ$ on $\F_T$, a  process
  $(\tS(t))_{t\in[0,T]}$ taking values in $\cO$, such that $\tS$ is a   $\PQ$
  martingale, bounded  in $L^2(\PQ)$ and finally  $\abs{S(t)-\tS(t)}\leq
  \eps_t$ almost surely for all $t\in[0,T]$. 
\end{theorem}

The main theorem of \cite{MR2398764}, covers the case when $\cO=(0,\infty)^n$
and $-\eta S_i(t)/(1+\eta)\leq \tS_i(t)-S_i(t)\leq \eta S_i(t)$ for
$i=1,\dots,n$, with some positive
constant $\eta>0$. That is, we get their result by the choice 
\begin{equation}
  \label{eq:eps_t}
  \eps_t= \frac{\eta}{1+\eta} \min_{i}{S_i(t)}  
\end{equation}
and \eqref{eq:eps-LS} holds with
$L_s=\eta$. So Theorem  \ref{thm:cps} contains the result of \citet{MR2398764}
as a special case. 
We also have to mention the recent paper of \citet{Maris2012}. They prove a
similar statement with $\eps_t=\eps$ constant.

Our Theorem \ref{thm:CPS} also follows easily from Theorem
\ref{thm:cps}; the  choice of $\eps_t$ given in \eqref{eq:eps_t} yields an
$\eta$-consistent price system evolving in $\cO$.

\begin{proof}[Proof of Theorem \ref{thm:cps}]To keep the process $\tS$ inside
  $\cO$ we decrease $\eps_t$, if neccesary, such that 
  \begin{equation}
    \label{eq:eps_min_d(S,O^c)}
    0<\eps_t<\inf\set*{\abs*{S_t-x}}{x\notin\cO},\quad
    \text{holds for all $t\in[0,T]$}.
  \end{equation} 
  Indeed, taking $\beps_t=\eps_t\wedge
  \frac12\inf\set{\abs*{S_t-x}}{x\notin\cO}$ the process $\beps$ is positive
  and fulfills \eqref{eq:eps-LS} with $\bL=L\vee (1/2)$. So in what follows we
  assume that \eqref{eq:eps_min_d(S,O^c)} holds also.

  The proof is based on two steps. First, similarly to the proof in
  \cite{MR2398764}, a 
  random walk with retirement is constructed. The properties of this random
  walk are collected in the next Lemma.
  \begin{lemma}\label{l2}
    Under the assumption of Theorem  \ref{thm:cps} there is
    a sequence of stopping times $(\tau_k)_{k\geq 1}$, a sequence
    of random variables $(X_k)_{k\geq0}$ and an equivalent probability $\PQ$
    such that 
    \begin{rlist}
    \item\label{l2:it1} $\tau_0=0$, $(\tau_k)$ is increasing and
      $\cup_k\event{\tau_k=T}$ has full probability,
      \smallskip
    \item\label{l2:it2} $(X_k)_{k\geq0}$ is a $\PQ$ martingale in the discrete
      time filtration $(\G_k=\F_{\tau_k})_{k\geq0}$, bounded in $L^2(\PQ)$, 
      \smallskip
    \item\label{l2:it3} if $\tau_k\leq t\leq \tau_{k+1}$ then
      $\abs{S_t-X_{k+1}}\leq\eps_t$. 
    \end{rlist}
  \end{lemma}

  The second step of the argument is to take $\tS_t=\E[\PQ]{X\cond \F_t}$, where
  $X=\lim_{k\to\infty} X_k$. 
  Then $\tS$ is 
  a martingale under $\PQ$ bounded in $L^2(\PQ)$ since the variable $X$
  is in $L^2(\PQ)$. 

  It remains to show that $\abs{\tS_t-S_t}\leq \eps_t$ for $t\in[0,T]$. By
  \eqref{eq:eps_min_d(S,O^c)} this ensures also that $\tS_t\in\cO$.
  By   the right continuity of $\tS$ and $S$ it is enough to deal with $t>0$.
  For $t\in(0,T]$ introduce 
  the  random index
  $\nu=\nu(t)=\inf\set{k}{\tau_k\geq t}$. Note that $\nu$ is almost surely
  finite   by (\ref{l2:it1}) of Lemma \ref{l2}. Clearly $\nu$ is 
  a $\G$ stopping time, and $\tau_{\nu}$ is a stopping time 
  in the filtration $\F$. For the stopped $\sigma$-fields we have
  $\G_\nu=\F_{\tau_\nu}$. 

  Then  $\tau_{\nu-1}\leq t\leq  \tau_{\nu}$ by the definition of $\nu$. As
  $(X_k)_{k\geq0}$ is a martingale, by property (\ref{l2:it2}), in the
  filtration $\G$ we have 
  $X_{\nu}=\E[\PQ]{X\cond \G_\nu}=\E[\PQ]{X\cond \F_{\tau_{\nu}}}$.
  By property (\ref{l2:it3})  $\abs{S_t-X_{\nu}}\leq
  \eps_t$. 
  Putting all these together, we get
  \begin{displaymath}
    \abs{S_t-\tS_t}=\abs*{S_t-\E[\PQ]{X\cond \F_t}}=
    \abs*{S_t-\E[\PQ]{X_{\nu(t)}\cond \F_t}}\leq
    \E[\PQ]{\abs*{X_{\nu(t)}-S_t}\cond \F_t}\leq \eps_t.\qedhere
  \end{displaymath}
\end{proof}

We use the next corollary of the conditional full support property, which also
justifies the name. It is related to the strong conditional full support in
the terminology of \cite{MR2398764}. We give at the end of this section a
direct proof instead of referring to the indirect proof using measurable
selection of \cite{MR2398764}.
 
\begin{corollary}\label{cor:scfs}
  Assume that the continuous adapted process $S$ evolving in $\cO$ has
  conditional full support in $\cO$. Let $\tau$ be a stopping time and denote
  by $Q_{S|\F_\tau}$ the regular version of the conditional distribution of
  $S$ given $\F_\tau$. 

  Then  the support of the random measure $Q_{S|\F_\tau}$ is 
  \begin{displaymath}
    \supp Q_{S|\F_\tau}=\set*{g\in
      C([0,T],\cO)}{g|_{[0,\tau]}=S|_{[0,\tau]}},\quad\text{almost surely}.
  \end{displaymath}
\end{corollary}

\begin{proof}[Proof of Lemma \ref{l2}]
  Without loss of generality we may assume that $\eps$ is decreasing. Indeed,  
  by taking $\beps_t=\min_{s\leq t} \eps_s$, we have $0< \beps_t\leq
  \eps_t$ and  for $s\leq t$ 
  \begin{displaymath}
    \abs*{\beps_t-\beps_s}\leq \sup_{s\leq u\leq t} \abs*{\eps_u-\eps_s}\leq
    L_s \sup_{s\leq u\leq t} \abs*{S_u-S_s}
  \end{displaymath}
  So the condition \eqref{eq:eps-LS} holds for $\beps$ as well.

  So in what follows we assume that $(\eps_t)_{t\in[0,T]}$ is decreasing and
  \eqref{eq:eps_min_d(S,O^c)} holds.

  The definition of $(\tau_k,X_k)_{n\geq0}$ is then straightforward. We take 
  $\tau_0=0$ and $X_0=S_0$. If $(\tau_k,X_k)$ are already defined then we take
  \begin{align*}
    \tau_{k+1}&=T \wedge \inf\set*{t>\tau_k}{\abs*{S_t-S_{\tau_k}}>\eps_t/2},\\
    X_{k+1}&=X_k\I{\tau_{k+1}=T}+S_{\tau_{k+1}}\I{\tau_{k+1}<T}.
  \end{align*}
  Now, it is easily seen from the definition 
  that $\abs*{X_{k+1}-S_{\tau_k}}\leq \eps_{\tau_{k+1}}/2.$ 
  Indeed, there are three cases 
  \begin{nlist}
  \item $\tau_{k+1}<T$, then $X_{\tau_{k+1}}=S_{\tau_{k+1}}$ and the
    estimation follows by the choice of $\tau_{k+1}$.
  \item $\tau_k<T=\tau_{k+1}$, then $X_{k+1}=X_k=S_{\tau_k}$ and the
    estimation is obvious.
  \item $\tau_k=\tau_{k+1}=T$. Then there is $k_0<k$ such that
    $\tau_{k_0}<\tau_{k_0+1}=T$, and
    $X_{k+1}=X_k=\cdots=X_{k_0}=S_{\tau_{k_0}}$ and
    $S_{\tau_k}=S_T=S_{\tau_{k_0+1}}$. By the choice of $\tau_{k_0+1}$ we have
    that $\abs{X_{k+1}-S_{\tau_k}}=\abs{S_{\tau_{k_0+1}}-S_{\tau_{k_0}}}\leq
    \eps_T/2=\eps_{\tau_{k+1}}/2$. 
  \end{nlist}

  Then for $\tau_k\leq t\leq \tau_{k+1}$
  \begin{displaymath}
    \abs*{X_{k+1}-S_t}\leq
    \abs*{X_{k+1}-S_{\tau_k}}+\abs*{S_t-S_{\tau_k}}\leq
    \frac12\zjel{\eps_{\tau_{k+1}}+\eps_t}\leq \eps_t
  \end{displaymath}
  as $\eps$ is decreasing.
  Hence Property \eqref{l2:it3} holds.

  Property \eqref{l2:it1} follows easily from the continuity of the sample
  path of $S$ on $[0,T]$. Indeed, assume that for a given $\omega$, we have
  $\tau_k(\omega)<T$ for all $k$. Then at $\tau(\omega)=\sup_k\tau_k(\omega)$
  the sample path $S(\omega)$ could not be continuous, as
  $\abs{S_{\tau_{k+1}}(\omega)-S_{\tau_k}(\omega)}\geq \eps_T(\omega)/2$.

  To construct the probability measure $\PQ$ and prove Property \eqref{l2:it2}
  we apply the argument of  \citet{MR2398764}. 
  With the notation $\Delta_{k+1}=X_{k+1}-X_{k}$
  they showed that if 
  \begin{align}
    \label{eq:int conv supp}
    &0\in\interior\conv\supp Q_{\Delta_{k+1}|\F_{\tau_k}},\quad
    \text{almost surely on  $\tau_{k+1}<T$}\\ 
    \label{eq:tau_n=T}
    &\P{\tau_{k+1}=T\cond \F_{\tau_k}}>0. 
  \end{align}
  then there exists an equivalent probability $\PQ$ satifying the requirements
  of the statement. 

  Roughly speaking, \eqref{eq:int conv supp}  implies the
  existence of $Z_k$ such that $\E{Z_k\Delta_{k}\cond \F_{\tau_{k-1}}}=0$ and
  $\E{Z_k\cond \F_{\tau_{k-1}}}=1$. One can define $Z_k$ in such a way that it
  charges most of the mass to the events $\event{\tau_k=T}$. With this it is
  possible to achieve that 
  \begin{equation}
    \label{eq:ZDelta}
   \E{Z_k\abs{\Delta_k}^2\cond[\Big] \F_{\tau_{k-1}}}\leq 2^{-k}, 
  \end{equation}
  and that the partial products $L_k=\prod_{\ell=1}^k Z_\ell$ are 
  convergent in $L^1$.
  Then with $L=\prod Z_k$ 
  and $\d\PQ= L \d\P$, using \eqref{eq:ZDelta} one can show  that $X\in
  L^2(\PQ)$. 
  For details we refer the reader to the proof of Theorem 1.2 in \cite[pages
  508-510]{MR2398764}. 
  
  So to finish the proof we have to show \eqref{eq:int conv supp} and
  \eqref{eq:tau_n=T}. Note, that it is enough to elaborate the proof for $k=0$
  and $\F_0$ being trivial, as by conditioning on $\F_{\tau_k}$ and
  linearly relabelling the time interval $[\tau_k,T]$ into $[0,T]$ we
  can reduce the 
  general case to this special case. Indeed all our arguments are based on
  full support of the conditional law of $S$ given $\F_{\tau_k}$, the
  properties of $(S,\eps)$ given in \eqref{eq:eps-LS},
  \eqref{eq:eps_min_d(S,O^c)} and the non-increase of $\eps$. Each of these
  hold under the regular version of the conditional law of $(S,\eps)$ given
  $\F_{\tau_k}$ and  they are not sensitive to a continuous
  time-change. 
  Even though the time-change is random it depends only
  on $\tau_k$, which is measurable with respect to
  $\F_{\tau_k}$.

  For \eqref{eq:tau_n=T} it is enough to show that $\P{\tau_1=T}>0$.
  Given $L_{0},\eps_{0}$ we
  define $\eta=\eps_{0}/(3(L_{0}\vee1))$. Then
  \begin{equation}\label{eq:small change}
   \P{\forall t\in[0,T],\,\abs{S_t-S_{0}}<\eta}>0,
  \end{equation}
  by Corollary \ref{cor:scfs}. Using condition \eqref{eq:eps-LS} we have
  \begin{displaymath}
    \eps_t\geq \eps_{0}-L_{0}\eta\geq\frac23\eps_{0}\geq 2\eta>
    2\abs{S_t-S_{0}},
    \quad 
    \text{for all $t\in[0,T]$ when $\sup_{t\in[0,T]}\abs{S_t-S_{0}}<\eta$
    }.
  \end{displaymath}
  That is
  \begin{displaymath}
   \event*{\sup_{t\in[0,T]}\abs{S_t-S_{0}}<\eta}
   \subset\event*{\tau_{1}=T}.  
  \end{displaymath}
  Hence \eqref{eq:tau_n=T} follows by \eqref{eq:small change} in the special
  case $k=0$ and $\F_{\tau_0}$ being trivial, and also in the general case as we
  have already remarked in the previous paragraph.

  Next we turn to \eqref{eq:int conv supp}. For the special case $k=0$ and
  $\F_{\tau_0}$ trivial, it simplifies to (by a slight abuse of notation)  
  \begin{equation}\label{eq:intconvsupp n=0}
    0\in \interior\conv \supp (X_1-X_0). 
  \end{equation}
  So we prove \eqref{eq:intconvsupp n=0}, from this  the general case follows. 

  Let us denote by $\pi_{r}$ the
  projection onto the ball with center $0$ and radius $r$, that is
  \begin{displaymath}
    \pi_{r}(y)=
    \begin{cases}
      \frac{r}{\abs{y}}y & \abs{y}\geq r,\\
      y &\text{otherwise}.
    \end{cases}
  \end{displaymath}
  We show below that there is  
  a positive $\delta$ such that
  $\pi_\delta\zjel{\supp (X_1-X_0)}$ contains the entire sphere
  $\set{y}{\abs{y}=\delta}$. This clearly implies
  \eqref{eq:intconvsupp n=0}.

  Let   $v$ be a unit vector in $\real^n$ and $\eta>0$. Define 
  \begin{displaymath}
    G_\eta=
    \set*{g\in C([0,T],\cO)}{\abs{g(t)-(S_0+v\eps_0t/T)}<\eta,\, t\in [0,T]}.   
  \end{displaymath}
  By \eqref{eq:eps_min_d(S,O^c)} $\eps_0$ is smaller than the
  distance of $S_0$ from the complement of $\cO$, hence $G_\eta$ is a non-empty
  open subset  of $C([0,T],\cO)$. By the full support property $\P{S\in
    G_\eta}>0$ for all $\eta>0$. 

  When $S\in G_\eta$ and $\eta$ is smaller
  than $\eps_0/2$, then $S$ exits the ball with center $S_0$ and
  radius $\eps_0/2$ hence $\tau_1<T$. At $\tau_1$ we have that 
   $\abs{X_1-X_0}=\abs{S_{\tau_1}-S_0}=\eps_{\tau_1}/2$ and 
   by \eqref{eq:eps-LS}
   \begin{displaymath}
     \eps_{\tau_1}\geq
     \eps_0-L_0\sup_{0\leq u\leq\tau_1}\abs{S_u-S_0}
   \end{displaymath}
   Define $f(t)=S_0+v\eps_0 t/T$.  Since $S\in G_\eta$
   \begin{align*}
     \sup_{0\leq u\leq\tau_1}\abs{S_u-S_0}&\leq 
     \eta+\sup_{0\leq u\leq\tau_1}\abs{f(u)-S_0}\\&=
     \eta+\abs{f(\tau_1)-S_0}\leq 2\eta+\abs{S_{\tau_1}-S_0}\\ &=
     2\eta+\frac12\eps_{\tau_1},
   \end{align*}
   that is
   \begin{displaymath}
     \eps_{\tau_1}\geq\eps_0-L_0(2\eta+\frac12\eps_{\tau_1}),\quad
     \eps_{\tau_1}\geq\frac{\eps_0-L_02\eta}{1+L_0/2}
   \end{displaymath}
  and
  \begin{displaymath}
    \abs{X_1-X_0}=\frac{\eps_{\tau_1}}2\geq\frac{\eps_0-L_02\eta}{2+L_0}.
  \end{displaymath}
  Now, taking $\eta$ so small that $2L_0\eta<\eps_0/2$ and
  $\delta=\eps_0/(4+2L_0)$ we obtain that the closed set $\pi_\delta(\supp
  (X_1-X_0))$ intersects the set 
  $\set{y}{\abs{y}=\delta,\quad \abs{y-\delta v}\leq \eta}$. Since this is
  true for all $\eta$ small enough and unit vector $v$ \eqref{eq:intconvsupp
    n=0} follows and the proof is complete.
\end{proof}

\begin{proof}[Proof of Corollary \ref{cor:scfs}]For a fixed open $G\subset C([0,T],\cO)$ the process $M_t=\P{S\in
    G\cond \F_t}$ is a non-negative martingale. We may take the c\`adl\`ag version of
  this martingale. Let $A=\event{M_\tau=0}$. Then by the martingale property
  $(M_t-M_{t\wedge\tau})\I[A]=\int_0^t \I[A]\I{\tau<s}dM_s$ is a non-negative
  martingale starting from zero, hence 
  \begin{equation}
    \label{eq:MA}
    \P{\forall t\in[0,T],\,(M_t-M_{t\wedge\tau})\I[A]=0}=1.
  \end{equation}

  We use the notation $G_t=\set{g|_{[0,t]}}{g\in G}$ and note that as $S$ has
  conditional full support in the sense of Definition \ref{df:cfs} we have
  that $M_t>0$ on the event $S|_{[0,t]}\in G_t$.

  Next we approximate $\tau$ by stopping times
  $\tau_k=2^{-k}\zjel{\br{2^k\tau}+1}$ and note that by \eqref{eq:MA}
  \begin{displaymath}
    A\cap\event{\tau_k=\ell 2^{-k}}\subset \event{M_{\ell2^{-k}}=0},
  \end{displaymath}
  holds up to a  null set.
  Thus
  \begin{displaymath}
    A^c\cap\event{\tau_k=\ell2^{-k}}\supset 
    \event{M_{\ell2^{-k}}>0}\cap\event{\tau_k=\ell2^{-k}}\supset 
    \event{S|_{[0,\ell2^{-k}]}\in G_{\ell2^{-k}}}\cap\event{\tau_k=\ell2^{-k}}
  \end{displaymath}
  Taking union we obtain
  \begin{displaymath}
    A^c\supset \event{S|_{[0,\tau_k]}\in G_{\tau_k}},
  \end{displaymath}
  and
  \begin{displaymath}
    A^c\supset \bigcup_{k=1}^\infty \event{S|_{[0,\tau_k]}\in G_{\tau_k}}=
    \event{S|_{[0,\tau]}\in G_{\tau}}, 
  \end{displaymath}
  where in the last step we used that $\tau_k$ approaches $\tau$ from the right
  and $G$ is open. On the other hand $\event{S|_{[0,\tau]}\notin
    G_\tau}\subset A$ is obvious so we can conclude that the events
  \begin{displaymath}
    A^c=\event{M_\tau=\P{S\in G\cond \F_\tau}>0},\quad\text{and}\quad
    \event{S|_{[0,\tau]}\in G_\tau} 
  \end{displaymath}
  are equal up to a negligible event. 
  
  Since the space $C([0,T],\cO)$ is second countable, its topology has 
  a countable base 
  and we may conclude that there is $\Omega'$ of full
  probability such that on $\Omega'$ for
      all open $G\subset C([0,T],\cO)$
  \begin{equation*}
    Q_{S|\F_\tau}(G)>0,\quad\text{exactly when $S|_{[0,\tau]}\in G_\tau$}. 
  \end{equation*}
  But then for  $\omega\in \Omega'$ the support of the random Borel measure
  $Q_{S|\F_\tau}(.,\omega)$ 
  on $C([0,T],\cO)$ is the (random) closed set
  $\set{g}{g|_{[0,\tau(\omega)]}=S(\omega)|_{[0,\tau(\omega)]}}$ as stated.
\end{proof}

\section{Conditional full support; 
  extension of a result of Pakkanen} 
\label{sec:CFS}
 
In this section we prove Theorem \ref{thm:CFS}. 
That is, we 
give a sufficient condition for a multidimensional
continuous semimartingale to have full support. 
It also gives the conditional full support of the process. As we
already remarked 
it seems to be new in the sense, that we do not assume that our process
solves a stochastic differential equation as it is done in \cite{MR2190038}
and it is not only for one dimensional processes as it is in
\cite{MR2731340}. We use comparison with a squared Bessel process of suitably
chosen dimension. 

\begin{theorem}\label{thm:cfs} 
  Let $X$ be a $d$-dimensional \ito/ process, such that 
  \begin{displaymath}
    \d X_t=\mu_t\d t+\sigma_t \d W_t
  \end{displaymath}
  Assume that $\abs{\mu}$, $\norm{a}$ and $\norm{a^{-1}}$ are  bounded
  processes, where $a_t=\sigma_t\sigma_t'$.
  Then $X$ has full support. 
\end{theorem}

The conditional full support of $X$, that is Theorem \ref{thm:CFS},
follows from the observation that 
$(X_u)_{u\in[s,T]}$ under the regular version of its conditional law given
$\F_s$ is an \ito/ process 
on the time interval $[s,T]$ satisfying the assumptions of Theorem \ref{thm:cfs}.
 
As it is observed in \cite{MR2731340}, under the assumption of Theorem
\ref{thm:cfs}, it is enough to
consider the case when $\mu=0$, $X_0=0$  and 
$G$ is an open ball around the identically zero function. This is the content
of Proposition \ref{prop:new} below. When $\mu=0$ and $X_0=0$ the process $X$ is a
martingale starting from zero, whose quadratic variation 
process $\q{X}_t=\int_0^t a_t\d t$ satisfies
\begin{equation}\label{eq:a}
  c I_d\leq a_t=\frac{\d\q{X}_t}{\d t} \leq C I_d, \quad \text{for all $t\geq0$}.
\end{equation}
Here $I_d$ is the identity matrix of dimension $d$.

We can also assume that $X$ is defined on $[0,\infty)$ although in Theorem
\ref{thm:cfs} it is defined only on $[0,T]$. We can simply extend it
using an independent $d$ dimensional Brownian motion $B$, by the formula
\begin{displaymath}
  \tX_t=
  \begin{cases}
    X_t&\text{if $t\leq T$},\\
    c'B_{t-T}+X_T&\text{if $t>T$}
  \end{cases}
\end{displaymath}
where $c'\in [c,C]$ is arbitrary.

\begin{proposition}\label{prop:new}
  Assume that for the $d$--dimensional continuous martingale
  $(X_t,\F_t)_{t\geq0}$ with $X_0=0$ \eqref{eq:a} holds. Then for all
  $T>0$ and $\eps>0$ 
  \begin{displaymath}
    \P{\sup_{t\leq T} \abs{X_t}\leq \eps}>0.
  \end{displaymath}
\end{proposition}
\begin{proof} 
  Let $R_t=\abs{X_t}^2=\sum_{i=1}^d (X^{(i)}_t)^2$. Then
  \begin{displaymath}
    R_t=2\int_0^t \sqrt{R_s}\frac{X^T_s\d X_s}{\abs{X_s}}+\text{Tr}(\q{X}_t)
  \end{displaymath}
  By \eqref{eq:a} 
  \begin{displaymath}
    N_t=\int_0^t \frac{X^T_s\d X_s}{\abs{X_s}}
  \end{displaymath}
  is a one dimensional martingale with $c\leq \frac{\d\q{N}_t}{\d t}\leq C$.
  Let $(\eta(t))_{t\geq 0}$ be the time change making $N$ a one-dimensional
  Brownian motion $\beta_t=N_{\eta(t)}$. For the time changed process
  $\tR_t=R_{\eta(t)}$ we have $R_t=\tR_{\q{N}_t}$ and  
  \begin{displaymath} 
    \tR_t=R_{\eta(t)}=2\int_0^t \sqrt{\tR_s}\d\beta_s+ \int_0^t b(s)\d s 
  \end{displaymath}
  Here $b(s)=\text{Tr}(a_{\eta(s)}) \frac{\d\eta(s)}{\d s}$ 
  and $0\leq b(s)\leq Cd/c$.

  Now we compare $\tR$ to 
  the solution of 
  \begin{displaymath}
    \d Z_t=2\sqrt{Z_t}\d\beta_t+\delta \d t, \quad Z_0=0.
  \end{displaymath}
  where $\delta\geq dC/c$. Note that the drift of $Z$ is bigger than that of
  $\tR$.  
  
  Then by a standard comparison result of the solutions of SDEs, $\tR_t\leq
  Z_t$ for all $t\geq 0$, see \cite[chapter IX, (3.7) Theorem on page
  394.]{revuz-yor}. 
  This argument is based on two simple
  observation. First, $\tR-Z$ can not accumulate local time at level zero, and
  then by Tanaka formula $\abs{\tR-Z}_+$ is a non-negative continuous
  semimartingale starting
  from zero having non-positive drift. This is only possible if
  $\abs{\tR-Z}_+$ is identically zero.  

  Now, $Z$ is a squared Bessel process of dimension $\delta$, so it stays
  below $\eps^2$ on $[0,s]$ with positive probability for any $\eps>0$ and
  $s\geq 0$. Since, $R_t=\tR_{\q{N}_t}$  and $\q{N}_t\leq CT$ for $t\leq T$ 
  we have that
  \begin{multline*}
    \P{\sup_{s\leq T}\abs{X_s}<\eps}= \P{\sup_{s\leq T} R_s<\eps^2}\\\geq
    \P{\sup_{s\leq CT} \tR_s<\eps^2}\geq
    \P{\sup_{s\leq CT} Z_s<\eps^2}>0.\qedhere
  \end{multline*}
\end{proof}


\begin{thebibliography}{12}

\bibitem[\protect\citeauthoryear{Delbaen and Schachermayer}{2006}]{MR2200584}
\begin{bbook}[author]
\bauthor{\bsnm{Delbaen},~\bfnm{Freddy}\binits{F.}} \AND
  \bauthor{\bsnm{Schachermayer},~\bfnm{Walter}\binits{W.}}
(\byear{2006}).
\btitle{The mathematics of arbitrage}.
\bseries{Springer Finance}.
\bpublisher{Springer-Verlag}, \baddress{Berlin}.
\bmrnumber{MR2200584 (2007a:91001)}
\end{bbook}
\endbibitem

\bibitem[\protect\citeauthoryear{Fernholz}{1999}]{MR1684221}
\begin{barticle}[author]
\bauthor{\bsnm{Fernholz},~\bfnm{Robert}\binits{R.}}
(\byear{1999}).
\btitle{On the diversity of equity markets}.
\bjournal{J. Math. Econom.}
\bvolume{31}
\bpages{393--417}.
\bdoi{10.1016/S0304-4068(97)00018-9}.
\bmrnumber{1684221 (2000e:91072)}
\end{barticle}
\endbibitem

\bibitem[\protect\citeauthoryear{Fernholz}{2001}]{MR1861997}
\begin{barticle}[author]
\bauthor{\bsnm{Fernholz},~\bfnm{Robert}\binits{R.}}
(\byear{2001}).
\btitle{Equity portfolios generated by functions of ranked market weights}.
\bjournal{Finance Stoch.}
\bvolume{5}
\bpages{469--486}.
\bdoi{10.1007/s007800100044}.
\bmrnumber{1861997 (2002h:91052)}
\end{barticle}
\endbibitem

\bibitem[\protect\citeauthoryear{Fernholz, Karatzas and
  Kardaras}{2005}]{MR2210925}
\begin{barticle}[author]
\bauthor{\bsnm{Fernholz},~\bfnm{Robert}\binits{R.}},
  \bauthor{\bsnm{Karatzas},~\bfnm{Ioannis}\binits{I.}} \AND
  \bauthor{\bsnm{Kardaras},~\bfnm{Constantinos}\binits{C.}}
(\byear{2005}).
\btitle{Diversity and relative arbitrage in equity markets}.
\bjournal{Finance Stoch.}
\bvolume{9}
\bpages{1--27}.
\bdoi{10.1007/s00780-004-0129-4}.
\bmrnumber{2210925 (2006k:60123)}
\end{barticle}
\endbibitem

\bibitem[\protect\citeauthoryear{Fernholz and Karatzas}{2009}]{Karatzas200989}
\begin{bincollection}[author]
\bauthor{\bsnm{Fernholz},~\bfnm{Robert}\binits{R.}} \AND
  \bauthor{\bsnm{Karatzas},~\bfnm{Ioannis}\binits{I.}}
(\byear{2009}).
\btitle{Stochastic Portfolio Theory: an Overview}.
In \bbooktitle{Handbook of Numerical Analysis},
(\beditor{\bfnm{P.~G.}\binits{P.~G.}~\bsnm{Ciarlet}}, ed.).
\bseries{Handbook of Numerical Analysis}
\bvolume{15}
\bpages{89 - 167}.
\bpublisher{Elsevier}.
\bdoi{10.1016/S1570-8659(08)00003-3}
\end{bincollection}
\endbibitem

\bibitem[\protect\citeauthoryear{Guasoni, R{\'a}sonyi and
  Schachermayer}{2008}]{MR2398764}
\begin{barticle}[author]
\bauthor{\bsnm{Guasoni},~\bfnm{Paolo}\binits{P.}},
  \bauthor{\bsnm{R{\'a}sonyi},~\bfnm{Mikl{\'o}s}\binits{M.}} \AND
  \bauthor{\bsnm{Schachermayer},~\bfnm{Walter}\binits{W.}}
(\byear{2008}).
\btitle{Consistent price systems and face-lifting pricing under transaction
  costs}.
\bjournal{Ann. Appl. Probab.}
\bvolume{18}
\bpages{491--520}.
\bdoi{10.1214/07-AAP461}.
\bmrnumber{2398764 (2009a:91050)}
\end{barticle}
\endbibitem

\bibitem[\protect\citeauthoryear{Guasoni and R\'asonyi}{2012}]{GR}
\begin{barticle}[author]
\bauthor{\bsnm{Guasoni},~\bfnm{Paolo}\binits{P.}} \AND
  \bauthor{\bsnm{R\'asonyi},~\bfnm{Mikl\'os}\binits{M.}}
(\byear{2012}).
\btitle{Fragility of Arbitrage and Bubbles in Diffusion Models}.
\bdoi{10.2139/ssrn.1856223}
\end{barticle}
\endbibitem

\bibitem[\protect\citeauthoryear{Maris and Sayit}{2012}]{Maris2012}
\begin{barticle}[author]
\bauthor{\bsnm{Maris},~\bfnm{Florian}\binits{F.}} \AND
  \bauthor{\bsnm{Sayit},~\bfnm{Hasanjan}\binits{H.}}
(\byear{2012}).
\btitle{Consistent Price Systems in Multiasset Markets}.
\bjournal{International Journal of Stochastic Analysis}
\bvolume{2012}
\bpages{14 pages}.
\bnote{Article ID 687376}.
\bdoi{10.1155/2012/687376}
\end{barticle}
\endbibitem

\bibitem[\protect\citeauthoryear{Osterrieder and
  Rheinl\"ander}{2006}]{Osterrieder2006}
\begin{barticle}[author]
\bauthor{\bsnm{Osterrieder},~\bfnm{J\"org~R.}\binits{J.~R.}} \AND
  \bauthor{\bsnm{Rheinl\"ander},~\bfnm{Thorsten}\binits{T.}}
(\byear{2006}).
\btitle{Arbitrage opportunities in diverse markets via a non-equivalent measure
  change}.
\bjournal{Annals of Finance}
\bvolume{2}
\bpages{287--301}.
\bdoi{10.1007/s10436-006-0037-z}
\end{barticle}
\endbibitem

\bibitem[\protect\citeauthoryear{Pakkanen}{2010}]{MR2731340}
\begin{barticle}[author]
\bauthor{\bsnm{Pakkanen},~\bfnm{Mikko~S.}\binits{M.~S.}}
(\byear{2010}).
\btitle{Stochastic integrals and conditional full support}.
\bjournal{Journal of Applied Probability}
\bvolume{47}
\bpages{650--667}.
\bdoi{10.1239/jap/1285335401}.
\bmrnumber{MR2731340}
\end{barticle}
\endbibitem

\bibitem[\protect\citeauthoryear{Revuz and Yor}{1991}]{revuz-yor}
\begin{bbook}[author]
\bauthor{\bsnm{Revuz},~\bfnm{Daniel}\binits{D.}} \AND
  \bauthor{\bsnm{Yor},~\bfnm{Marc}\binits{M.}}
(\byear{1991}).
\btitle{Continuous martingales and {B}rownian motion}.
\bseries{Grundlehren der Mathematischen Wissenschaften [Fundamental Principles
  of Mathematical Sciences]}
\bvolume{293}.
\bpublisher{Springer-Verlag}, \baddress{Berlin}.
\bmrnumber{MR1083357 (92d:60053)}
\end{bbook}
\endbibitem

\bibitem[\protect\citeauthoryear{Stroock and Varadhan}{2006}]{MR2190038}
\begin{bbook}[author]
\bauthor{\bsnm{Stroock},~\bfnm{Daniel~W.}\binits{D.~W.}} \AND
  \bauthor{\bsnm{Varadhan},~\bfnm{S.~R.~Srinivasa}\binits{S.~R.~S.}}
(\byear{2006}).
\btitle{Multidimensional diffusion processes}.
\bseries{Classics in Mathematics}.
\bpublisher{Springer-Verlag}, \baddress{Berlin}.
\bnote{Reprint of the 1997 edition}.
\bmrnumber{2190038 (2006f:60005)}
\end{bbook}
\endbibitem

\end{thebibliography}

\end{document}